\documentclass[12pt]{article}

\makeatletter
\def\input@path{{styles/}}
\makeatother

\def\UseBibLatex{1}

\ifx\sarielComp\undefined%
\newcommand{\SarielComp}[1]{}
\newcommand{\NotSarielComp}[1]{#1}%
\else
\newcommand{\SarielComp}[1]{#1}%
\newcommand{\NotSarielComp}[1]{}%
\fi
\newcommand{\IfPrinterVer}[2]{#2}%

\newcommand{\UsePackage}[1]{%
  \IfFileExists{../styles/#1.sty}{%
      \usepackage{../styles/#1}%
   }{%
      \IfFileExists{./styles/#1.sty}{%
         \usepackage{styles/#1}%
      }{%
         \usepackage{#1}%
      }%
   }%
}

\usepackage[cm]{fullpage}%
\usepackage{amsmath}%
\usepackage{amssymb}%
\usepackage{bm}%
\usepackage{caption}
\usepackage{xcolor}%
\usepackage{upgreek}%
\usepackage{euscript}%
\usepackage{caption}%

\usepackage{mathcalb}%

\usepackage{marvosym}

\SarielComp{\usepackage{sariel_colors}}%

\usepackage[amsmath,thmmarks]{ntheorem}%
\theoremseparator{.}%

\usepackage{titlesec}%
\titlelabel{\thetitle. }%
\usepackage{xcolor}%
\usepackage{mleftright}%
\usepackage{xspace}%
\usepackage{scalerel}%

\usepackage[inline]{enumitem}

\newlist{compactenumA}{enumerate}{5}%
\setlist[compactenumA]{topsep=0pt,itemsep=-1ex,partopsep=1ex,parsep=1ex,%
   label=(\Alph*)}%

\newlist{compactenumi}{enumerate}{5}%
\setlist[compactenumi]{topsep=0pt,itemsep=-1ex,partopsep=1ex,parsep=1ex,%
   label=(\roman*)}%

\usepackage{ifluatex}
\usepackage{ifxetex}

\ifluatex %
      \usepackage{fontspec}
      \usepackage[utf8]{luainputenc}
\else
       \ifxetex %
          \usepackage{fontspec}
       \else
          \usepackage[T1]{fontenc}
          \usepackage[utf8]{inputenc}
       \fi
\fi

\providecommand{\BibLatexMode}[1]{}
\providecommand{\BibTexMode}[1]{#1}

\ifx\UseBibLatex\undefined%
  \renewcommand{\BibLatexMode}[1]{}
  \renewcommand{\BibTexMode}[1]{#1}
\else
  \renewcommand{\BibLatexMode}[1]{#1}
  \renewcommand{\BibTexMode}[1]{}
\fi

\BibLatexMode{%
   \usepackage[bibencoding=utf8,style=alphabetic,backend=biber]{biblatex}%
   \usepackage{sariel_biblatex}%
}

\usepackage{hyperref}%

\IfPrinterVer{%
   \usepackage{hyperref}%
}{%
   \usepackage{hyperref}%
   \hypersetup{%
      breaklinks,%
      ocgcolorlinks, colorlinks=true,%
      urlcolor=[rgb]{0.25,0.0,0.0},%
      linkcolor=[rgb]{0.5,0.0,0.0},%
      citecolor=[rgb]{0,0.2,0.445},%
      filecolor=[rgb]{0,0,0.4},
      anchorcolor=[rgb]={0.0,0.1,0.2}%
   }
}

\definecolor{blue25}{rgb}{0,0,0.7}
\providecommand{\emphic}[2]{%
   \textcolor{blue25}{%
      \textbf{\emph{#1}}}%
   \index{#2}}
\providecommand{\emphi}[1]{\emphic{#1}{#1}}

\definecolor{almostblack}{rgb}{0, 0, 0.3}

\providecommand{\emphw}[1]{}%
\renewcommand{\emphw}[1]{{\textcolor{almostblack}{\emph{#1}}}}%

\theoremseparator{.}%

\theoremstyle{plain}%
\newtheorem{theorem}{Theorem}[section]

\newtheorem{claim}[theorem]{Claim}%

\newtheorem{observation}[theorem]{Observation}

\theoremstyle{definition}%
\theoremheaderfont{\sf} \theorembodyfont{\upshape}%
\newtheorem*{remark:unnumbered}[theorem]{Remark}%

\newtheorem{defn}[theorem]{Definition}

\newcommand{\myqedsymbol}{\rule{2mm}{2mm}}

\theoremheaderfont{\em}%
\theorembodyfont{\upshape}%
\theoremstyle{nonumberplain}%
\theoremseparator{}%
\theoremsymbol{\myqedsymbol}%
\newtheorem{proof}{Proof:}%

\newcommand{\atgen}{\symbol{'100}}
\newcommand{\SarielThanks}[1]{\thanks{Department of Computer Science;
      University of Illinois; 201 N. Goodwin Avenue; Urbana, IL,
      61801, USA; {\tt sariel\atgen{}illinois.edu}; {\tt
         \url{http://sarielhp.org/}.} #1}}

\newcommand{\MitchellThanks}[1]{%
   \thanks{%
      Department of Computer Science;
      University of Illinois; 201 N. Goodwin Avenue; Urbana, IL,
      61801, USA; {\tt mfjones2\atgen{}illinois.edu}; {\tt
         \url{http://mfjones2.web.engr.illinois.edu/}.} #1}}

\numberwithin{figure}{section}%
\numberwithin{table}{section}%
\numberwithin{equation}{section}%

\newcommand{\HLink}[2]{\hyperref[#2]{#1~\ref*{#2}}}
\newcommand{\HLinkSuffix}[3]{\hyperref[#2]{#1\ref*{#2}{#3}}}

\newcommand{\figlab}[1]{\label{fig:#1}}
\newcommand{\figref}[1]{\HLink{Figure}{fig:#1}}

\newcommand{\apndlab}[1]{\label{apnd:#1}}

\newcommand{\obslab}[1]{\label{observation:#1}}
\newcommand{\obsref}[1]{\HLink{Observation}{observation:#1}}
\newcommand{\obsrefY}[2]{\hyperref[observation:#1]{#2}}

\newcommand{\clmlab}[1]{\label{claim:#1}}
\newcommand{\clmref}[1]{\HLink{Claim}{claim:#1}}

\providecommand{\eqlab}[1]{}%
\renewcommand{\eqlab}[1]{\label{equation:#1}}

\providecommand{\remove}[1]{}%
\renewcommand{\remove}[1]{}%

\newcommand{\pth}[2][\!]{\mleft({#2}\mright)}%

\newcommand{\cardin}[1]{\left| {#1} \right|}%

\renewcommand{\Re}{\mathbb{R}}%

\DefineNamedColor{named}{OliveGreen} {cmyk}{0.44,0,0.95,0.90}

\newcommand{\eps}{\varepsilon}

\newcommand{\expandY}[2]{{#1_{\oplus #2}}}
\newcommand{\norm}[1]{\left\| {#1} \right\|}

\newcommand{\permut}[1]{\left\langle {#1} \right\rangle}%
\newcommand{\DotProd}[2]{\permut{{#1},{#2}}}%

\newcommand{\distY}[2]{\left\| {#1} - {#2} \right\|}
\newcommand{\dY}[2]{\left\| {#1} - {#2} \right\|}

\newlength{\savedparindent}
\newcommand{\SaveIndent}{\setlength{\savedparindent}{\parindent}}
\newcommand{\RestoreIndent}{\setlength{\parindent}{\savedparindent}}

\providecommand{\Mh}[1]{#1}
\renewcommand{\Mh}[1]{#1}

\newcommand{\pp}{\Mh{{p}}}

\providecommand{\Mh}[1]{#1}

\newcommand{\dSY}[2]{\Mh{\mathsf{d}}\pth{#1,#2}}%
\newcommand{\nnY}[2]{\Mh{\mathsf{n}}_{#2}\pth{#1}}%
\newcommand{\nnX}[1]{{#1}_\Mh{\mathsf{n}}}%
\newcommand{\hsX}[1]{\Mh{h}_{#1}}%
\newcommand{\body}{\Mh{\mathsf{C}}}
\newcommand{\bodyA}{\Mh{\mathsf{D}}}

\newcommand{\Line}{\mathsf{k}}%

\newcommand{\ang}{\ensuremath{\text{\Anglesign}}}

\BibLatexMode{%
   \bibliography{dudley_approx} }

\begin{document}

\title{A Proof of Dudley's Convex Approximation}%

\author{Sariel Har-Peled%
   \SarielThanks{Work on this paper was partially supported by a NSF
      AF awards CCF-1421231, and %
      CCF-1217462.  %
   }%
   \and%
   Mitchell Jones%
   \MitchellThanks{}%
}

\date{\today}

\maketitle

\begin{abstract}
    We provide a self contained proof of a result of Dudley
    \cite{d-mescs-74}, which shows that a bounded convex-body in
    $\Re^d$ can be $\eps$-approximated, by the intersection of
    $O_d\bigl(\eps^{-(d-1)/2} \bigr)$ halfspaces, where $O_d$ hides
    constants that depends on $d$.
\end{abstract}

\section{Statement and proof}
\apndlab{convex:approx}

For a closed convex body $\body \subseteq \Re^d$, let
$\expandY{\body}{\eps}$ denote the set of all points in $\Re^d$ in
distance at most $\eps$ from $\body$. In particular,
$\body \subseteq \expandY{\body}{\eps}$, and the Hausdorff distance
between $\body$ and $\expandY{\body}{\eps}$ is $\eps$.  For a point
$p \in \Re^d$, the \emphi{projection} of $p$ to $\body$ is the point
$\nnX{p} = \nnY{p}{\body} = \arg \min_{x \in X} \dY{p}{x}$. Thus, the
\emphw{distance} of $p$ from $\body$ is
$\dSY{p}{\body} = \dY{p}{\nnX{p}}$.

\begin{defn}
    For a closed convex set $\body \subseteq \Re^d$, a vector
    $v \neq 0$ is a \emphw{normal} at $p$ to $\body$, if
    $\DotProd{v}{p} = \max_{ q \in \body} \DotProd{v}{q}$, where
    $\DotProd{v}{p}$ denotes the dot-product of $v$ and $p$.
\end{defn}

\begin{observation}
    \obslab{same}%
    For any point $u \in \Re^d \setminus \body$, we have that
    $u - \nnX{u}$ is normal at $\nnX{u}$ to $\body$. Furthermore, for
    all $t\geq 0$, the nearest-neighbor to $\nnX{u} + t(u -\nnX{u})$
    on $\body$ is $\nnX{u}$.
\end{observation}

\begin{observation}[Projection into a convex set is a contraction]
    \obslab{contraction}%
    For a compact convex set $X \subseteq \Re^d$, and any points
    $p,q \in \Re^d$, let $p' = \nnY{p}{X}$ and $q' = \nnY{q}{X}$. We
    have $\dY{p'}{ q'} \leq \dY{p}{q}$. Indeed, consider the segment
    $s = p'q' \subseteq X$, and let $\Line$ be the line spanned by
    $s$, and observe that
    $\dY{p'}{ q'} = \dY{\nnY{p}{s}}{\nnY{q}{s}} \leq
    \dY{\nnY{p}{\Line}}{\nnY{q}{\Line}} \leq \dY{p}{q}$.
\end{observation}

\begin{theorem}[\cite{d-mescs-74}]
    For $d \geq 2$, let $\body$ be a closed convex body in $\Re^d$,
    such that $\body$ is contained in the ball of radius $1$ centered
    at the origin.  For a parameter $\eps \in (0,1)$, one can compute
    a convex body $\bodyA$, which is the intersection of
    $O_d(1/\eps^{(d-1)/2})$ halfspaces, such that
    $\body \subseteq \bodyA \subseteq \expandY{\body}{\eps}$, where
    $O_d$ hides a constant exponential in $d$.
\end{theorem}
\begin{proof}
    Let $S$ be the sphere of radius $3$ centered at the origin, and
    let $Q$ be a maximal $\delta$-packing of $S$, where
    $\delta = \sqrt{\eps}$. We remind the reader that a set
    $Q \subseteq S$ is a \emphi{$\delta$-packing} if \smallskip%
    \begin{compactenumi}
        \smallskip%
        \item for any point $u \in S$, there is a point $v \in Q$,
        such that $\dY{u}{v}\leq \delta$, and

        \smallskip%
        \item for any two points $u,v \in Q$, we have that
        $\dY{u}{v} \geq \delta$.
    \end{compactenumi}
    \smallskip%
    It is easy to verify that
    $\cardin{Q} = \Theta_d(1/ \delta^{d-1}) =
    \Theta_d(\eps^{-(d-1)/2})$, see \cite[Section 6.3]{b-fpc-04}.  For
    a point $u \in \Re^d \setminus \body$, consider the halfspace
    $\hsX{u}$ containing $\body$, whose boundary passes through
    $\nnX{u}$ and is orthogonal to the vector $u- \nnX{u}$.

    \begin{figure}%
        \centerline{\includegraphics{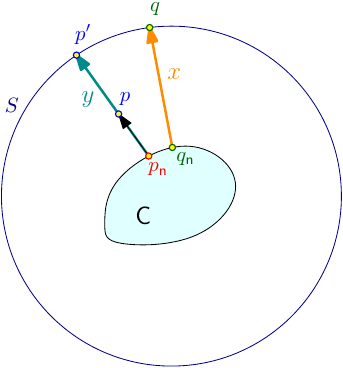}}
        \caption{}
        \figlab{project}
    \end{figure}

    It is clear that
    $\body \subseteq \bodyA = \bigcap_{q \in Q} \hsX{q}\Bigr.$. As for
    the other direction, consider any point
    $\pp \in \partial \bodyA \setminus \body$. We claim that
    $\dSY{p}{\body} \leq \eps$.  The ray $\nu$ emanating from
    $\nnX{p}$ in the direction of $\pp - \nnX{p}$ hits $S$ at a point
    $p'$. Let $q \in Q$ be the nearest point in the $\delta$-packing
    $Q$ to $p'$, see \figref{project}.  By \obsref{same}, we have
    $\nnX{p'} = \nnX{p}$. This implies that
    \begin{math}
        \distY{ \nnX{p}}{ \nnX{q}}%
        \leq%
        \distY{p'}{q}%
        \leq%
        \delta,
    \end{math}
    as projection into a convex set is a
    \obsrefY{contraction}{contraction}.  Thus, for $x = q - \nnX{q}$
    and $y = p' - \nnX{p}$, we have
    \begin{equation*}
        \displaystyle%
        \distY{x}{y} %
        =%
        \norm{\bigl.  q - \nnX{q} - \pth{ p' - \nnX{p}} } =
        \norm{\bigl.  \pth{\nnX{p} - \nnX{q}} + \pth{q - p'} }%
        \leq
        \norm{\nnX{p} - \nnX{q}} + \norm{q - p'}%
        \leq
        2 \delta.
    \end{equation*}

    \SaveIndent%
    \vspace{-0.5cm}%

    \begin{figure}[h]
        \centering%
        {\includegraphics[page=1]{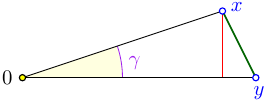}}%
        \captionof{figure}{The triangle formed by $x,y$ and the
           origin.}%
        \figlab{x:y}
    \end{figure}
    Observe that %
    \begin{math}
        \norm{x} = \| q - \nnX{q}\| \geq \dSY{\body}{S} \geq 2.
    \end{math}
    Let $\gamma$ be the angle between $x$ and $y$.  The distance of
    $x$ from the line spanned by $y$ is
    $\norm{x} \sin\gamma \leq \dY{x}{y}$, see \figref{x:y}.  Thus, we
    have
    \begin{equation*}
        \sin \gamma%
        \leq%
        \frac{\dY{x}{y}}{\norm{x}}
        \leq%
        \frac{2\delta}{2}
        \leq
        \delta.
    \end{equation*}

    If $\nnX{p} \in \partial \hsX{q}$ (or $\nnX{p} =\nnX{q}$) then
    $p = \nnX{p}$, and $\dSY{p}{\body}=0$ (if this is not clear, see
    \clmref{a} below). A similar argument shows that
    $\gamma \leq \pi/3$, see \figref{x:y:2}.

    As $\nnX{p} \neq \nnX{q}$, and $q \in S$, it follows that they are
    all distinct, and not all lying on a common line (if the later is
    not clear, see \clmref{b} below).

    \noindent%
    \begin{minipage}{0,58\linewidth}
        \medskip%
        \RestoreIndent%

        Let $f$ be the two dimensional plane that contains the points
        $\nnX{q},q$, and $\nnX{p}$, see \figref{small}.  Let $p''$ be
        the projection of $\nnX{p}$ to $\partial \hsX{q}$. Observe
        that $p'' - \nnX{p} = \alpha x$, for some $\alpha > 0$, and
        thus $p'' \in f$.  If $f \subseteq \partial \hsX{q}$ then
        $\nnX{p} \in \partial \hsX{q}$, and we are done.  Thus, the
        set $f \cap \partial \hsX{q}$ is a line, denoted by
        $\Line$. This line contains $\nnX{q}$ and $p''$.  Observe that
        $p''$ is in the interior of $\Re^d \setminus \hsX{p}$ and
        $\nnX{q} \in \hsX{p}$, and thus
        $p'' \nnX{q} \cap \partial{\hsX{p}} \neq \emptyset$.  Let $t$
        be the closest point in $p'' \nnX{q} \cap \partial{\hsX{p}}$
        to $\nnX{q}$ (i.e., $t$ might be $\nnX{q}$).

    \end{minipage}
    \hfill
    \begin{minipage}{0.42\linewidth}
        \captionof{figure}{The plane $f$.}
        \figlab{small}%
        \vspace*{-1.0cm}%
        \centerline{\includegraphics{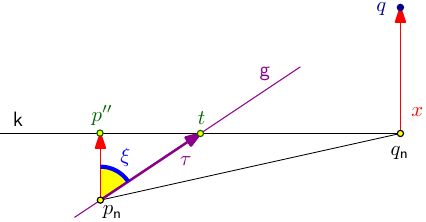}}

    \end{minipage}

    \bigskip%
    Let $\tau = t - \nnX{p}$, and consider the angle
    $\xi = \ang(\tau,x)$.  Observe that $\tau$ is orthogonal to $y$
    (as both its endpoints are in $\partial \hsX{p}$).  Namely,
    $0 \leq \ang(y,x) \leq \gamma$ and $\ang(\tau,y) = \pi/2$.
    Applying the triangle inequality to angles, we have
    \begin{equation*}
        \xi
        =%
        \ang(\tau,x)%
        \geq
        \ang(\tau,y) - \ang(y,x) %
        \geq%
        \pi/2 - \gamma.\Bigr.
    \end{equation*}
    Inspecting \figref{small}, implies that $\xi \in [0,\pi/2]$,
    implying that $\xi \in[\pi/2-\gamma, \pi/2]$.  Thus, we have
    \begin{equation*}
        \dSY{p}{\body}
        =%
        \dY{p}{\nnX{p}}%
        \leq
        \dSY{\nnX{p}}{\partial \hsX{q}}
        \leq%
        \dY{\nnX{p}}{ p'' } %
        \leq%
        \norm{\nnX{p} - \nnX{q}} \cos \xi%
        \leq%
        \delta \sin \gamma%
        \leq%
        \delta^2
        \leq %
        \eps.
    \end{equation*}
\end{proof}

\subsection{More details}

Here, we elaborate over some details in the above proof that might not
be obvious to the reader.

\noindent%
\begin{minipage}{0.45\linewidth}
    \centering
    \includegraphics{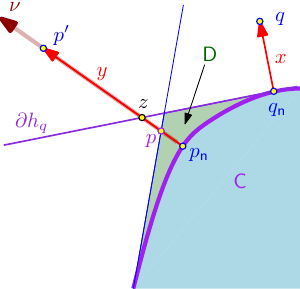}
    \captionof{figure}{}
    \figlab{ray}
\end{minipage}
\hfill
\begin{minipage}{0.45\linewidth}
    \centering%
    \includegraphics[page=3]{figs/x_y}%
    \captionof{figure}{The angle $\gamma$ is minimized if $x$ and
       $y$ are as short as possible (i.e., $2$), and $\dY{x}{y}$
       is maximized (i.e., $2\delta \leq 2$).}
    \figlab{x:y:2}
\end{minipage}

    \begin{claim}
    \clmlab{a}%
    If $\nnX{p} \in \partial \hsX{q}$ then $\dSY{p}{\body}=0$.
\end{claim}
\begin{proof}
    If $p \in \body$, then there is nothing prove, so assume
    $p \notin C$.  By definition,
    $\dSY{p}{\body} = \dY{p}{\nnX{p}} > 0$.  By the packing property,
    the angle between $x$ and $y$ is at most $\pi/3$, see
    \figref{x:y:2}. By construction,
    $\nnX{p}, \nnX{q} \in \partial \body \subseteq \hsX{q}$.  The ray
    $\nu$ emanating from $\nnX{p}$ in the direction of $y$, contains
    $p$.  Because $\ang(x,y) \leq \pi/3$, the ray $\nu$ hits
    $\partial \hsX{q}$ at a point $z = \nu \cap \partial \hsX{q}$, see
    \figref{ray}.  The point $p$ lies on the segment $ z
    \nnX{p}$. Thus, if $p_n \in \partial \hsX{q}$, then the segment
    $z \nnX{p}$ is a point, and then $z=\nnX{p}$ and $\nnX{p}=p$. This
    implies that $\dSY{p}{\body} = \dY{p}{\nnX{p}} = 0$.
\end{proof}

\begin{figure}[h!]
    \centering%
    \includegraphics{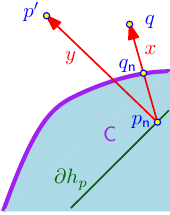}
    \caption{}
    \figlab{on:line}
\end{figure}

\begin{claim}
    \clmlab{b}%
    The points $\nnX{p}, \nnX{q}, q$ do not all lie on a common line.
\end{claim}
\begin{proof}
    The points $\nnX{p}, \nnX{q}$, and $q$ are all distinct.  If these
    three points lie on a common line (say in this order), see
    \figref{on:line}, then $\nnX{q}$ lies in the interior of
    $\Re^d \setminus \hsX{p}$, since the angle between the two normals
    at $\nnX{p}$ and $\nnX{q}$ (i.e., $y$ and $x$, respectively) is
    $\gamma$, which is at most $\pi/3$. This implies that
    $\\nnX{q} \notin \body$, which is impossible.

    The case that the points appear in the order $\nnX{q}, \nnX{p}$,
    and $q$ along the line is handled in a similar fashion.
\end{proof}

\paragraph*{Acknowledgement.}

We thank Michael Lesnick and Kenneth McCabe for pointing out issues in
earlier versions of this writeup, and providing detailed feedback that
significantly improved this writeup.

\BibTexMode{%
   \bibliographystyle{alpha}%
   \bibliography{dudley_approx}%
}

\BibLatexMode{\printbibliography}

\end{document}